\pdfoutput=1
\newif\ifFull
\Fullfalse

\documentclass{article}
\usepackage{amsthm,amsmath,amsfonts}
\usepackage{hyperref,graphicx,subcaption}
\usepackage[noend]{algorithmic}
\usepackage{cite}
\usepackage{url}

\graphicspath{{figures/}{data/}}

\newtheorem{lemma}{Lemma}

\newtheorem{theorem}{Theorem}

\begin{document}

\title{Parallel Equivalence Class Sorting: 
Algorithms, Lower Bounds, and Distribution-Based Analysis}

\author{
William E. Devanny \textsuperscript{1}\\[5pt]
Dept. of Computer Science \\
Univ. of California, Irvine \\
Irvine, CA 92697 USA \\
\texttt{wdevanny@uci.edu}
\and
Michael T.~Goodrich \\[5pt]
Dept. of Computer Science \\
Univ. of California, Irvine \\
Irvine, CA 92697 USA \\
\texttt{goodrich@uci.edu}
\and
Kristopher Jetviroj \\[5pt]
Dept. of Computer Science \\
Univ. of California, Irvine \\
Irvine, CA 92697 USA \\
\texttt{kjetviro@uci.edu}
}

\date{}

\maketitle

\begin{abstract}
We study parallel comparison-based algorithms for finding all 
equivalence classes of a set of $n$ elements, where sorting according
to some total order is not possible.
Such scenarios arise, for example, in applications, such as 
in distributed computer security, where each of $n$ agents are working to 
identify the private group to which they belong, with the only operation
available to them being a zero-knowledge pairwise-comparison (which is sometimes
called a ``secret handshake'') that reveals only whether
two agents are in the same group or in different groups.
We provide new parallel algorithms for this problem, as well as new
lower bounds and distribution-based analysis.
\end{abstract}

\footnotetext[1]{William E. Devanny was supported by an NSF Graduate Research Fellowship under grant DGE-1321846.}

\section{Introduction}

In the \emph{Equivalence Class Sorting} problem, we are given
a set, $S$, of $n$ elements and an equivalence relation, and we are asked
to group the elements of the set into their equivalence classes by only 
making pairwise equivalence tests
(e.g., see~\cite{Jayapaul2015}).
For example, 
imagine a convention of $n$ political interns where each
person at the convention belongs to one of $k$ political parties, 
such as Republican, Democrat, Green, Labor, Libertarian, etc., 
but no intern wants to openly
express his or her party affiliation unless they know they are 
talking with someone of their same party.  
Suppose further that each party has a secret handshake that
two people can perform that allows them to 
determine whether they are in the same political party (or 
they belong to different unidentified parties).
We are interested in this paper in the computational complexity of 
the equivalence class sorting problem in distributed and parallel settings,
where we would like to minimize the total number of parallel comparison
rounds and/or the total number of comparisons needed in order to classify
every element in $S$.

An important property of the equivalence class sorting problem is that it is
not possible to order the elements in $S$ according to some total ordering
that is consistent with the equivalence classes. Such a restriction could
come from a general lack of such an ordering or from security or privacy
concerns.  For example, consider the following applications:
\begin{itemize}
\item
\emph{Generalized fault diagnosis}. Suppose that each of $n$ different
computers are in one of $k$ distinct malware states, depending on whether
they have been infected with various computer worms. 
Each worm does not wish to reveal its presence, but it
nevertheless has an ability to detect when another computer 
is already infected with it (or risk autodetection by an exponential
cascade, as occurred with the Morris worm~\cite{Morris}).
But a worm on one computer 
is unlikely to be able to detect a different kind of worm
on another computer.
Thus, two computers can only compare each
other to determine if they have exactly the same kinds of infections or not. 
The generalized fault diagnosis problem, therefore, is to have the $n$
computers classify themselves into $k$ malware groups depending on their
infections, where the only testing method available is for two computers to
perform a pairwise comparison that tells 
them that they are either in the same malware
state or they are in different states.
This is a generalization of the classic fault diagnosis problem, where
there are only two states, ``faulty'' or ``good,'' which is studied in
a number of interesting papers, including one from the very first
SPAA conference (e.g., see 
\cite{Beigel:1989,Beigel:1993,b492587,Goodrich2008199,PU:46643,p4039201}).
\item
\emph{Group classification via secret handshakes}.
This is a cryptographic analogue to
the motivating example given above of interns at a political convention. 
In this case, $n$ agents are each assigned to one of $k$ 
groups, such that any two agents can perform a cryptographic 
``secret handshake'' protocol that results in them learning only whether
they belong to the same group or not
(e.g., see~\cite{Castelluccia2004,Jarecki2007,Sorniotti2010619,Xu:2004}).
The problem is to perform an efficient number of pairwise secret-handshake
tests in a few parallel rounds so that each agent identifies itself with 
the others of its group.
\item
\emph{Graph mining}.
Graph mining is the study of structure in collections of 
graphs~\cite{Cook:2006}. 
One of the algorithmic problems in this 
area is to classify which of a collection of
$n$ graphs 
are isomorphic to one another (e.g., see~\cite{Parthasarathy2010}).
That is, testing if two graphs are in the same group involves performing
a graph isomorphism comparison of the two graphs, which is a computation that
tends to be nontrivial but is nevertheless computationally feasible in some
contexts (e.g., see~\cite{graphs}).
\end{itemize}
Note that each of these applications contains two 
important features that
form the essence of the equivalence class sorting problem:
\begin{enumerate}
\item
In each application,
it is not possible to sort elements according to a known total order,
either because no such total order exists or because it would break 
a security/privacy condition to provide such a total order.
\item
The equivalence or nonequivalence between two 
elements can be determined only through pairwise comparisons. 
\end{enumerate}

There are nevertheless some interesting differences between
these applications, as well, which motivate
our study of two different versions of the equivalence class sorting problem.
Namely, in the first two applications, the comparisons done in any given
round in an algorithm must be disjoint, since the elements themselves
are performing the comparisons. In the latter two
applications, however, the elements are the objects of 
the comparisons, and
we could, in principle, allow for comparisons involving multiple 
copies of the same element in each round.
For this reason, we allow for two versions of the equivalence class sorting
problem:
\begin{itemize}
\item
\emph{Exclusive-Read (ER)} version. In this version, each element in $S$ can
be involved in at most a single comparison of itself and another
element in $S$ in any given comparison round.
\item
\emph{Concurrent-Read (CR)} version. In this version, each element in $S$ can
be involved in multiple comparisons of itself and other elements in $S$
in any comparison round.
\end{itemize}
In either version, we are interested in minimizing the number of parallel
comparison rounds and/or the total number of comparisons needed
to classify every element of $S$ into its group.

Because we expect the 
number parallel comparison rounds and the total number of comparisons
to be the main performance bottlenecks,
we are interested here in studying the equivalence class sorting problem in
Valiant's parallel comparison model~\cite{Valiant}, 
which only counts steps in which 
comparisons are made.
This is a synchronous computation
model that does not count any steps done between comparison steps, 
for example, to aggregate groups of equivalent elements based on 
comparisons done in previous steps.

\subsection{Related Prior Work}
In addition to the references cited above that motivate 
the equivalence class sorting problem or study the 
special case when the number of groups, $k$, is two,
Jayapaul {\it et al.}~\cite{Jayapaul2015} study the general 
equivalence class sorting problem,
albeit strictly from a sequential perspective.
For example, they show that one can solve the equivalence class sorting
problem using $O(n^2/\ell)$ comparisons, where $\ell$ is the size of the smallest 
equivalence class.
They also show that this problem has a lower bound of $\Omega(n^2/\ell^2)$ even
if the value of $\ell$ is known in advance.

The equivalence class sorting problem is, of course, 
related to comparison-based
algorithms for computing the majority or mode of a set of elements,
for which there is an extensive set of prior research
(e.g., see~\cite{Alonso1993253,Alonso2013495,DOBKIN1980255,ref11}).
None of these algorithms for majority or mode result in efficient parallel
algorithms for the equivalence class sorting problem, however.

\subsection{Our Results}
In this paper, we study the equivalence class sorting (ECS) 
problem from a parallel perspective, providing a 
number of new results, including the following:
\begin{enumerate}
\item
The CR version of the ECS problem can be solved in $O(k + \log\log n)$ 
parallel rounds using $n$ processors, were  $k$ is the number of equivalence classes.
\item
The ER version of the ECS problem can be solved in $O(k\log n)$ 
parallel rounds using $n$ processors, were  $k$ is the number of equivalence classes.
\item
The ER version of the ECS problem can be solved in $O(1)$ 
parallel rounds using $n$ processors, for the case when $\ell$ is at least $\lambda n$, for
a fixed constant $0<\lambda\le 0.4$,
where $\ell$ is the size of the smallest equivalence class.
\item
If every equivalence class is of size $f$, then solving the ECS problem
requires $\Omega(n^2/f)$ total comparisons. 
This improves a lower bound of $\Omega(n^2/f^2)$ by
Jayapaul {\it et al.}~\cite{Jayapaul2015}.
\item
Solving the ECS problem requires $\Omega(n^2/\ell)$ total comparisons,
where $\ell$ is the size of the smallest equivalence class.
This improves a lower bound of $\Omega(n^2/\ell^2)$ by
Jayapaul {\it et al.}~\cite{Jayapaul2015}.
\item
In Section~\ref{sec:sort-dists},
we study how to efficiently solve
the ECS problem when the input is drawn from a
known distribution on equivalence classes.  In this setting, we assume
$n$ elements have been sampled and fed as input to the algorithm.
We establish a relationship between the mean of the distribution
and the algorithm's total number of comparisons, 
obtaining upper bounds with high
probability for a variety of interesting distributions.  
\item
We provide the results of several
experiments to validate the
results from Section~\ref{sec:sort-dists} and study how total comparison
counts change as parameters of the distributions change.
\end{enumerate}
Our methods are based on several novel techniques, including 
a two-phased compounding-comparison technique for the parallel upper bounds and
the use of a new coloring argument for the lower bounds.


\section{Parallel Algorithms} \label{sec:para-alg}
In this section, we provide efficient parallel algorithms for solving
the equivalence class sorting (ECS) problem in Valiant's parallel
model of computation~\cite{Valiant}.
We focus on both the exclusive-read (ER) and concurrent-read (CR) versions 
of the problem, and
we assume we have $n$ processors, each of which
can be assigned to one equivalence comparison test to perform in a given
parallel round.  
Note, therefore, that any lower bound, $T(n)$, on the 
total number of comparisons needed to solve the ECS problem (e.g., as
given by Jayapaul {\it et al.}~\cite{Jayapaul2015}
and as we discuss in Section~\ref{sec:lower-bounds}), immediately implies
a lower bound of $\Omega(T(n)/n)$ for the number of parallel rounds
of computation using $n$ processors per round.
For instance,
these lower bounds imply that
the number of parallel rounds for solving the ECS problem with $n$ processors
must be $\Omega(n/\ell)$ and $\Omega(k)$, respectively,
where $k$ is the number of equivalence classes and $\ell$ is the size of the
smallest equivalence class.

With respect to upper bounds, recall that
Jayapaul {\it et al.}~\cite{Jayapaul2015}
studied the ECS problem from a sequential perspective.
Unfortunately, their algorithm cannot 
be easily parallelized, because the comparisons performed in a ``round'' of
their algorithm depend on the results from other comparisons in that same
round.
Thus, new parallel ECS algorithms are needed.

\subsection{Algorithms Based on the Number of Groups}
In this subsection, we describe CR and ER algorithms based on knowledge
of the number of groups, $k$.

If two sets of elements are sorted into their equivalence classes,
merging the two answers into the answer for the union requires at
most $k^2$ equivalence tests by simply performing a comparison
between every pair of equivalence class one from the first answer
and one from the second.  This idea leads to the following 
algorithm, which uses a two-phased compounding-comparison technique to 
solve the ECS problem:

\begin{enumerate}
\item Initialize a list 
of $n$ answers containing the individual input elements.
\item While the number of processors per answer is less than $4k^2$, 
merge pairs of answers by performing $k^2$ tests.
\item While there is more than one answer, let $ck^2$ be the number of processors available per answer and merge $c$ answers together by performing at most ${c \choose 2} k^2$ tests between each of the answers.
\end{enumerate} 

We analyze this algorithm in the following two lemmas
and we illustrate it in Figure~\ref{fig:algo-phases}.

\begin{figure*}[tbp]
\centering
\includegraphics[scale=0.8]{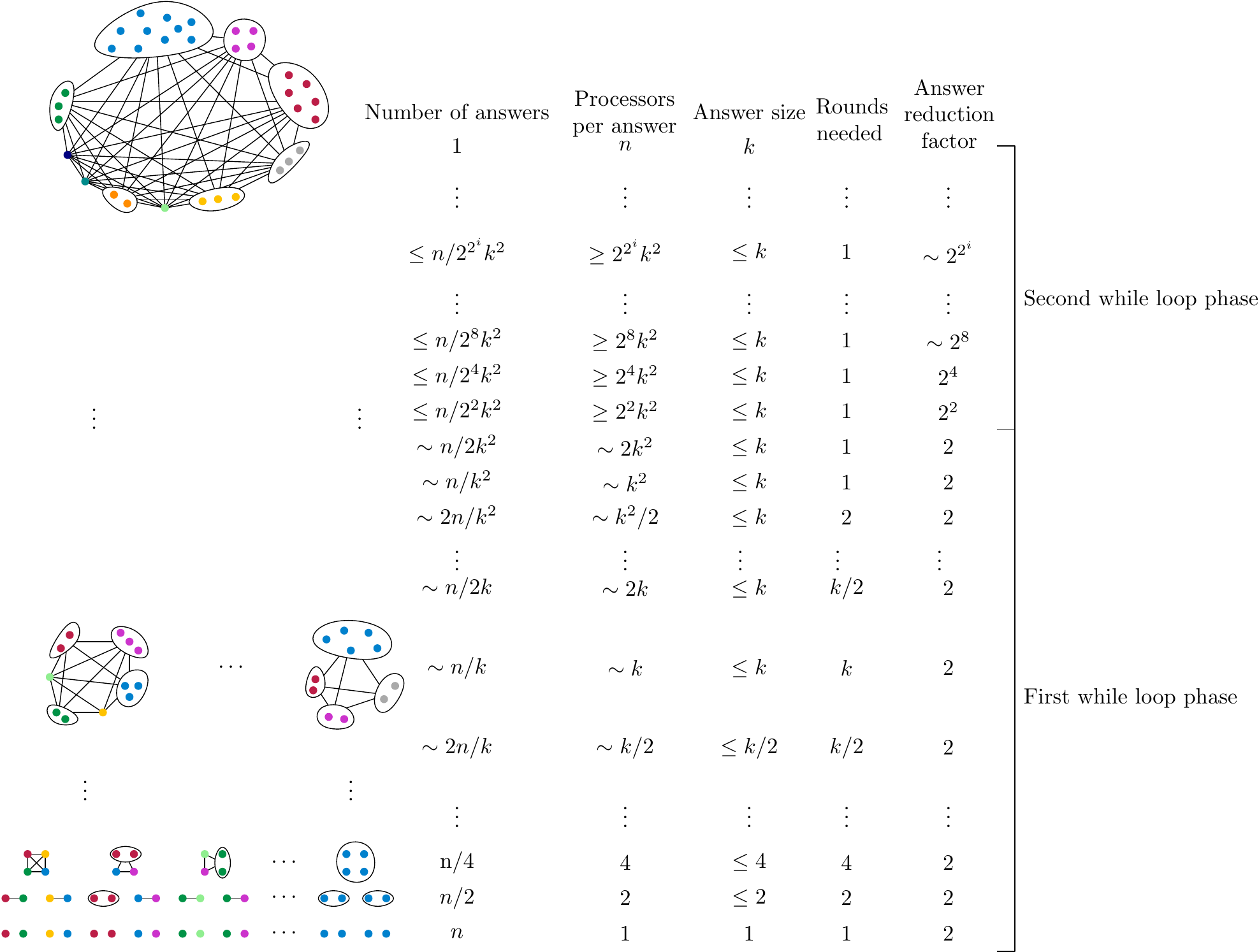}
\caption{A visualization of the parallel algorithm with a table on the right keeping track of relevant numbers for each loop iteration.}
\label{fig:algo-phases}
\end{figure*}

\begin{lemma}\label{lem:first-while}
The first while loop takes $O(k)$ rounds to complete.
\end{lemma}
\begin{proof}
In each round the number of equivalence classes in an answer at
most doubles until it reaches the upper bound of $k$.  In loop
iteration $i \leq \lceil \log k \rceil$, the answers are size at
most $2^i$ and there are $2^i$ processors per answer.  Therefore
it takes at most $2^i$ rounds to merge two answers.  The number of
rounds to reach the $\lceil \log k \rceil$ loop iteration is $O(k)$.
For loop iterations $\lceil \log k \rceil < i < \lceil \log k
\rceil^2$, the answers are size at most $k$, but there are still
at most $2^i$ processors per answer.  The number of rounds needed
for these iterations is also $O(k)$, as it forms a geometric sum that
adds up to be $O(k)$.
This part of the algorithm is illustrated in the bottom half of
Figure~\ref{fig:algo-phases}.
\end{proof}

\begin{lemma}\label{lem:second-while}
The second while loop takes $O(\log \log n)$ rounds to complete.
\end{lemma}
\begin{proof}
When entering the second while, there are more processors per answer
than needed to merge just two answers at a time.  If an answer has
access to $ck^2$ processors, then a group of ${c \choose 2}$ answers
can merge into one answer in a single round.  This means that if
there are $n/(ck^2)$ answers at the start of a round, then we merge
groups of $c^2/2$ answers into one answer and there are $n/(c^3k/2)$
answers remaining.  Because $c\geq 4$ by the condition of the first
while loop, in the iteration $i$ of the second while loop, there
are at most $n/(2^{2^i}k)$ answers.  And so the second while loop
will terminate after $O(\log \log n)$ rounds with the single answer
for the entire input.
This is illustrated in the top half of Figure~\ref{fig:algo-phases}.
\end{proof}

Combining these two lemmas, we get the following.

\begin{theorem}
The CR version of the equivalence class sorting problem on $n$ elements and $k$ equivalence classes can be solved in $O(k + \log \log n)$ parallel rounds of equivalence tests,
using $n$ processors in Valiant's parallel comparison model.
\end{theorem}
\begin{proof}
Lemmas~\ref{lem:first-while}~and~\ref{lem:second-while}.
\end{proof}

We also have the following.

\begin{theorem}
The ER version of the equivalence class sorting problem on $n$ elements and $k$ equivalence classes can be solved in $O(k\log n)$ parallel rounds of equivalence tests,
using $n$ processors in Valiant's parallel comparison model.
\end{theorem}
\begin{proof}
Merging two answers for the ER version of the ECS problem
model will always take at most $k$ rounds.  Repeatedly merging answers will arrive at one answer in $\log n$ iterations.  So equivalence class sorting can be done in $O(k\log n)$ parallel rounds of equivalence tests.
\end{proof}

\subsection{Algorithms Based on the Smallest Group Size}
In this subsection, we describe ER algorithms based on knowledge
of $\ell$, the size of the smallest equivalence class. We assume in this
section that $\ell\ge \lambda n$, for some constant $\lambda>0$, and we
show how to solve the ECS problem in this scenario using $O(1)$ 
parallel comparison rounds.
Our methods are generalizations of previous methods for 
the parallel fault diagnosis problem when there are only two classes, ``good'' 
and ``faulty''~\cite{Beigel:1989,Beigel:1993,b492587,Goodrich2008199}.
Let us assume, therefore, that there are at least 3 equivalence classes.

We begin with a theorem from Goodrich~\cite{Goodrich2008199}.


\begin{theorem}[Goodrich~\cite{Goodrich2008199}]
\label{thm-fault}
Let $V$ be a set of $n$ vertices, and let $0<\gamma,\lambda<1$.
Let $H_d=(V,E)$ be a directed graph
defined by the union of $d$ independent
randomly-chosen\footnote{That is, $H_d$ is defined by the
  union of cycles determined by $d$ random
  permutations of the $n$ vertices in $V$, so $H_d$ is, by definition,
  a simple directed graph.}
Hamiltonian cycles on $V$
(with all such cycles equally likely).
Then, for all subsets $W$ of $V$ of $\lambda n$ vertices,
$H_d$ induces at least one strongly connected component on $W$ of
size greater than $\gamma\lambda n$,
with probability at least
\[
1 - e^{n[(1+\lambda) \ln 2 + d(\alpha \ln \alpha + \beta \ln \beta
                              - (1-\lambda) \ln (1-\lambda))] + O(1)} ,
\]
where $\alpha=1-\frac{1-\gamma}{2} \lambda$
and $\beta=1-\frac{1+\gamma}{2} \lambda$.
\end{theorem}

In the context of the present paper,
let us take $\gamma=1/4$, so
$\alpha=1-(3/8)\lambda$ and $\beta=1-(5/8)\lambda$.
Let us also assume that $\lambda\le 0.4$, since we are considering the
case when the number of equivalence classes is at least $3$; hence, the smallest
equivalence class is of size at most $n/3$.

Unfortunately, using standard approximations for the natural
logarithm is not sufficient for us to employ the above probability bound
for small values of $\lambda$. 
So instead we use the following inequalities,
which hold for $x$ in the range $[0,0.4]$ (e.g., see~\cite{Kozma}),
and are based on the Taylor series for the natural logarithm:
\[
-x - \frac{x^2}{2} - \frac{x^3}{2} \le \ln (1-x) 
                           \le -x - \frac{x^2}{2} - \frac{x^3}{4}.
\]

These bounds allow us to bound the main term, $t$, in the above probability
of Theorem~\ref{thm-fault}
(for $\gamma=1/4$) as follows:
\begin{eqnarray*}
t &=&
\alpha \ln \alpha\ +\ \beta \ln \beta 
- (1-\lambda) \ln (1-\lambda) \\
&=&
(1-\frac{3}{8}\lambda) \ln (1-\frac{3}{8}\lambda)\ +\ (1-\frac{5}{8}\lambda) 
                       \ln (1-\frac{5}{8}\lambda) \\
&& -\ (1-\lambda) \ln (1-\lambda) \\
&\le& 
(1-\frac{3}{8}\lambda) \left(-\frac{3}{8}\lambda - 
    \frac{1}{2}\left(\frac{3}{8}\lambda\right)^2
    - \frac{1}{4}\left(\frac{3}{8}\lambda\right)^3\right) \\
&&+\  (1-\frac{5}{8}\lambda) \left(-\frac{5}{8}\lambda - 
    \frac{1}{2}\left(\frac{5}{8}\lambda\right)^2
    - \frac{1}{4}\left(\frac{5}{8}\lambda\right)^3\right) \\
&&-\ (1-\lambda) \left(-\lambda - 
    \frac{\lambda^2}{2} - \frac{\lambda^3}{2}\right) \\
&\le& -\frac{3743}{8192}\lambda^4 + \frac{19}{256} \lambda^3 - 
       \frac{15}{64} \lambda^2,
\end{eqnarray*}
which, in turn, is at most
\[
-\frac{\lambda^2}{8},
\]
for $0<\lambda\le 0.4$.
Thus, since this bound is negative 
for any constant $0<\lambda\le 0.4$, we can set $d$ to be a constant
(depending on $\lambda$) so that 
Theorem~\ref{thm-fault} holds with high probability.

Our ECS algorithm, then, is as follows:
\begin{enumerate}
\item
Construct a graph, $H_d$, as in Theorem~\ref{thm-fault},
as described above, with $d$ set 
to a constant so that
the theorem holds for the fixed $\lambda$ in the range
$(0,0.4]$ that is given.
Note that this step does not require 
any comparisons; hence, we do not count the time for this step in our
analysis (and the theorem holds with high probability in any case).
\item
Note that $H_d$ is a union of $d$ Hamiltonian cycles.
Thus, let us perform all the comparisons in $H_d$ in $2d$ rounds.
Furthermore, we can do this set of comparisons
even for the ER version of the problem.
Moreover, since $d$ is $O(1)$, this step involves a constant number of parallel
rounds (of $O(n)$ comparisons per round).
\item
For each strongly connected component, $C$, in $H_d$ consisting of elements
of the same equivalence class, compare the elements in $C$ with the other
elements in $S$, taking $|C|$ at a time.
By Theorem~\ref{thm-fault}, $|C|\ge \lambda n/8$. Thus, this step can 
be performed
in $O(1/\lambda)=O(1)$ rounds for each connected component; hence it
requires $O(1)$ parallel rounds in total.
Moreover, after this step completes, we will necessarily have identified
all the members of each equivalence class.
\end{enumerate}
We summarize as follows.

\begin{theorem}
Suppose $S$ is a set of $n$ elements,
such that the smallest equivalence class in $S$ is of size at least $\lambda n$,
for a fixed constant, $\lambda$, in the range $(0,0.4]$.
Then the
ER version of the equivalence class sorting problem on $S$ can be solved
in $O(1)$ parallel rounds using $n$ processors in Valiant's parallel comparison
model.
\end{theorem}

This theorem is true regardless of whether or not $\lambda$ is known.  If the value of $\lambda$ is not known, it is possible to repeatedly run the ECS algorithm starting with an arbitrary constant of $0.4$ for $\lambda$ and halving the constant whenever the algorithm fails.  Once the value is less than the unknown $\lambda$, the algorithm will succeed and the number of rounds will be independent of $n$ and a function of only the constant $\lambda$.

As we show in the next section, this performance is optimal when 
$\ell\ge \lambda n$, for a fixed constant $\lambda\in(0,0.4]$.

\section{Lower Bounds} \label{sec:lower-bounds}

The following lower bound questions were left open by 
Jayapaul {\it et al.}~\cite{Jayapaul2015}:

\begin{itemize}
\item 
If every equivalence class has size $f$, the 
total number of comparisons needed to solve
the equivalence class sorting problem 
$\Theta(n^2/f)$ or $\Theta(n^2/f^2)$?
\item 
Is the total number of 
comparisons
for finding an element in the smallest equivalence class $\Theta(n^2/\ell)$ or $\Theta(n^2/\ell^2)$?
\end{itemize}
Speaking loosely these lower bounds can be thought of as a question of how difficult it is for an element to locate its equivalence class.  The $\Theta(n^2/f)$ and $\Theta(n^2/\ell)$ bounds can be interpreted as saying the average element needs to compare to at least one element in most of the other equivalence classes before it finds an equivalent element.  Because there must be ${x \choose 2}$ comparisons between $x$ equivalence classes, the $\Theta(n^2/f^2)$ and $\Theta(n^2/\ell^2)$ bounds say we do not need too many more comparisons then the very minimal number needed just to differentiate the equivalence classes.  It seems unlikely that so few comparisons are required and we prove that this intuition is correct by proving lower bounds of $\Omega(n^2/f)$ and $\Omega(n^2/\ell)$ comparisons.

Note that these lower bounds are on the total number of comparisons needed to accomplish a task, that is they bound the work a parallel algorithm would need to perform.  By dividing by $n$, they also give simple bounds on the number of rounds needed in either the ER or CR models.

With respect to such lower bound questions as these,
let us maintain the state of an algorithm's
knowledge about element relationships in a simple graph.  At each
step, the vertex set of this graph is a partition of the elements
where each set is a partially discovered equivalence class for $S$.  
Thus, each element in $S$ is associated with exactly one vertex in this graph
at each step of the algorithm, and a vertex can have multiple elements
from $S$ associated with it.
If a pair of elements was compared and found to not be equal, then
there should be an edge in between the two vertices containing those
elements.  So initially the graph has a vertex for each element and
no edges.  When an algorithm tests equivalence for a pair of elements,
then, if the elements are not equivalent, the appropriate edge is
added (if it is absent) and, if the elements are equivalent, the two
corresponding vertices are contracted into a new vertex whose set
is the union of the two.  A depiction of this is shown in
\autoref{fig:equiv-test}.  An algorithm has finished sorting once
this graph is a clique and the vertex sets are the corresponding
equivalence classes.

\begin{figure*}[t]
\centering
\includegraphics[scale=0.8]{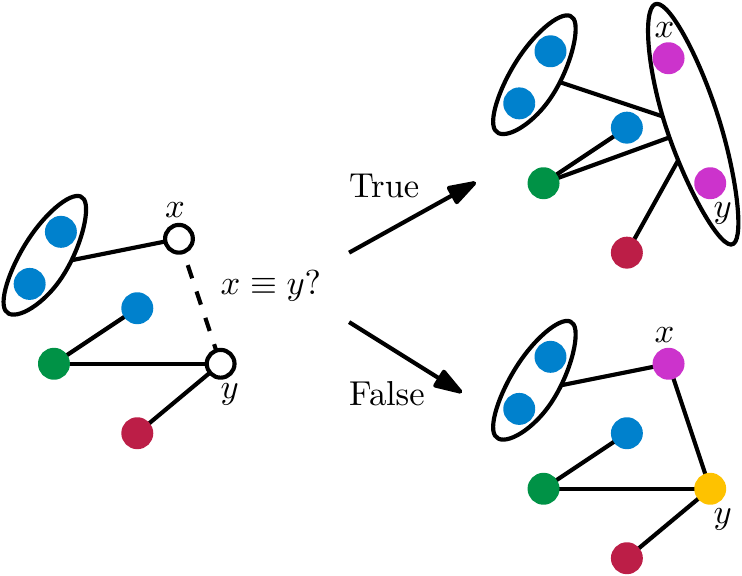}
\caption{We test if $x$ and $y$ are in the same equivalence class.  If they are, their vertices are contracted together.  If they are not, an edge is added.}
\label{fig:equiv-test}
\end{figure*}

An \emph{equitable $k$-coloring} of a graph is a proper coloring of a graph such that the size of each color class is either $\lfloor n/k \rfloor$ or $\lceil n/k \rceil$.  A \emph{weighted equitable $k$-coloring} of a vertex weighted graph is a proper coloring of a graph such that the sum of the weight in each color class is either $\lfloor n/k \rfloor$ or $\lceil n/k \rceil$.  Examples of these can be seen in \autoref{fig:equitable-colorings}.

\begin{figure*}[t]
\centering
\includegraphics[scale=0.8]{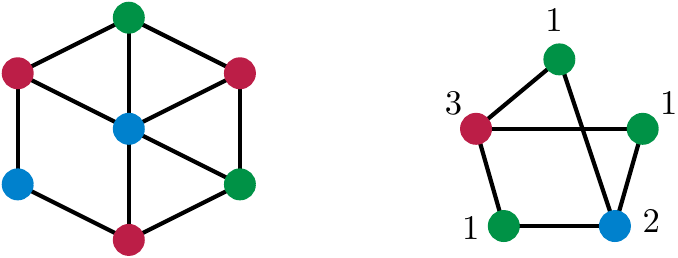}
\caption{On the left we have a graph with an equitable $3$-coloring and on the right we have a graph with a weighted equitable $3$-coloring.}
\label{fig:equitable-colorings}
\end{figure*}

An adversary for the problem of equivalence class 
sorting when every equivalence class has the same size $f$ (so $f$ divides $n$) must maintain that the graph has a weighted equitable $n/f$-coloring where the weights are the size of the vertex sets.  The adversary we describe here will maintain such a coloring and additionally mark the elements and the color classes in a special way.  It proceeds as follows.

First, initialize an arbitrary equitable coloring on the starting
graph that consists of $n$ vertices and no edges.  
For each comparison of two elements done by the adversary algorithm, let us
characterize how we react based on the following case analysis:
\begin{itemize}
\item 
If either of the elements is unmarked and this comparison would
increase its degree to higher than $n/4f$, then mark it as having 
``high'' element degree.

\item 
If either element is still unmarked, they currently have the same
color, and there is another unmarked vertex such that it is not
adjacent to a vertex with the color involved in the comparison and
no vertex with its color is adjacent to the unmarked vertex in the
comparison (i.e. we can have it swap colors with one of the vertices
in the comparison), then swap the color of that element and the unmarked
element in the comparison.

\item 
If either element is still unmarked, they currently have the same
color, and there is no other unmarked vertex with a different
unmarked color not adjacent to the color of the two elements being
compared, then mark all elements with the color involved in the
comparison as having ``high'' color degree and mark the color as having
``high'' degree.

\item At this point, 
either both elements are marked and we answer based on their color,
or one of the elements is unmarked and they have different colors,
so we answer ``not equal'' to the adversary algorithm.
\end{itemize}

\begin{figure*}[t]
\centering
\includegraphics[scale=1.2]{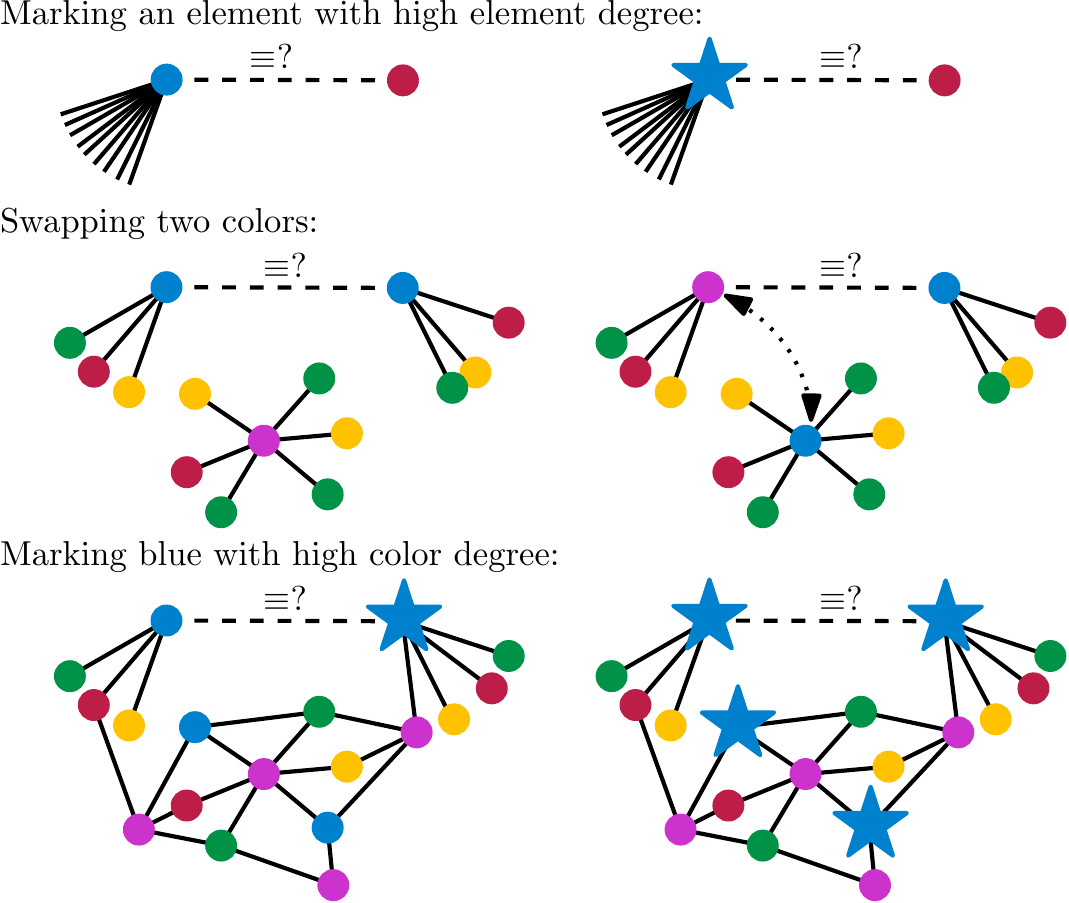}
\caption{Three cases of how the adversary works to mark vertices and swap colors.  The dashed line indicates the two elements being compared.  Marked vertices are denoted with stars.}
\label{fig:adversary}
\end{figure*}

At all times, the vertices that contain unmarked elements all have
weight one, because the adversary only answers equivalent for
comparisons once both vertices are marked.  When a color class is
marked, all elements in that color class are marked as having ``high''
color degree.  A few of the cases the adversary goes through are
depicted in Figure~\ref{fig:adversary}.

\begin{lemma}\label{lem:comp-count}
If $n/8$ elements are marked during the execution of an algorithm, then $\Omega(n^2/f)$ comparisons were performed.
\end{lemma}
\begin{proof}
There are three types of marked vertices: those with ``high'' element 
degree marks, those with ``high'' color degree marks, 
and those with both marks.  

The color classes must have been marked as having ``high'' degree
when a comparison was being performed between two elements of that
color class and there were no unmarked color candidates to swap
colors with.  Because one of the elements in the comparison had
degree less than $n/4f$, only a quarter of the elements have a color
class it cannot be swapped with.  So if there were at least $n-n/4$
unmarked elements in total, then the elements in the newly marked color
class must have been in a comparison $n/2$ times.

The ``high'' element degree 
elements were involved in at least $n/4f$ comparisons each.
So if $i$ color classes were marked and $j$ elements were 
only marked with ``high'' element degree, then 
the marked elements must have been a part of a test at least 
$ni/2 + nj/4f \geq (i + j/f)n/4$ times.  
Once $if + j \geq n/8$, then at least $n^2/64f$ equivalence tests 
were performed.
\end{proof}

\begin{theorem}
If every equivalence class has the same size $f$, then sorting requires at least $\Omega(n^2/f)$ equivalence comparisons.
\end{theorem}
\begin{proof}
When an algorithm finishes sorting, each vertex will have weight
$f$ and so the elements must all be marked.  Thus, by
\autoref{lem:comp-count}, at least $\Omega(n^2/f)$ comparisons must
have been performed.
\end{proof}

We also have the following lower bound as well.

\begin{theorem}
Finding an element in the smallest equivalence class, whose size is $\ell$, requires at least $\Omega(n^2/\ell)$ equivalence comparisons.
\end{theorem}
\begin{proof}
We use an adversary argument similar to the previous one, but we start
with $\ell$ vertices colored a special \emph{smallest class color
(scc)} and seperate the remaining $n-\ell$ vertices into $\lfloor(n
- \ell)/(\ell + 1)\rfloor$ color classes of size $\frac{n}{\lfloor(n
- \ell)/(\ell + 1)\rfloor}$ or $\frac{n}{\lfloor(n - \ell)/(\ell +
1)\rfloor} + 1$.

There are two changes to the previous adversary responses.  First,
the degree requirement for having ``high'' degree is now $n/4\ell$.
Second, if an scc element is about to be marked as having ``high'' degree,
we attempt to swap its color with any valid unmarked vertex.  Otherwise,
we proceed exactly as before.

If an algorithm attempts to identify an element as belonging to the
smallest equivalence class, no scc elements are marked, and
there have been fewer than $n/8$ elements marked, then the identified
element must be able to be swapped with a different color and the
algorithm made a mistake.  Therefore, to derive a lower bound 
for the total number of comparisons,
it suffices to derive a lower bound for the number of equivalence tests
until an scc element is marked.

The scc color class cannot be marked as having ``high'' color degree
until at least one scc element has high element degree.  However,
as long as fewer than $n/8$ elements are marked, we will never mark
an scc element with ``high'' degree.  So at least $n/8$ elements need
to be marked as having ``high'' element degree or ``high'' color degree and,
by the same type of counting as in Lemma~\ref{lem:comp-count},
$\Omega(n^2/\ell)$ equivalence tests are needed.  \end{proof}

\section{Sorting Distributions} \label{sec:sort-dists}

In this subsection, we study a version of the equivalence class sorting
problem where we are
given a distribution, $D$, on a countable set, $S$, 
and we wish to enumerate the set in order of most likely to 
least likely, $s_0, s_1, s_2,\dots$.  
For example, consider the following distributions:
\begin{itemize}
\item
Uniform: In this case, $D$ is a distribution on $k$ equivalence classes,
with each equivalence class being equally likely for every element of $S$.
\item
Geometric: Here, $D$ is a distribution such that the $i$th 
most probable equivalence class has probability $p^i(1-p)$.
Each element ``flips'' a biased coin where ``heads'' occurs with probability 
$p$ until it comes up ``tails.'' Then that element is in equivalence
class $i$ if it flipped $i$ heads.
\item
Poisson:
In this case, $D$ is model of the number of times an event occurs in
an interval of time, with an expected number of events determined by a 
parameter $\lambda$. 
Equivalence class $i$ is defined to be all the samples
that have the same number of events occurring, where the probability of
$i$ events occurring is 
\[
\frac{\lambda^i e^{-\lambda}}{i!}\ .
\]
\item
Zeta: 
This distribution, $D$, is related to Zipf's law, and models when the
sizes of the equivalence classes follows a power law, based on 
a parameter, $s>1$, which is common
in many real-world scenarios, such as the frequency of words in natural
language documents.
With respect to equivalence classes, the $i$th equivalence class 
has probability
\[
\frac{i^{-s}}{\zeta(s)},
\]
where $\zeta(s)$ is Riemann zeta function (which normalizes the probabilities
to sum to 1).
\end{itemize}

So as to number equivalence classes from most likely to least likely,
as $i=0,1,\ldots$, 
define $D_\mathbb{N}$ to be a distribution on the natural numbers such that 
\[
\Pr_{x\sim D_\mathbb{N}} \left[ x = i\right] = \Pr_{y\sim D} \left[ y = s_i\right].
\]
Furthermore,
so as to ``cut off'' this distribution at $n$,
define $D_\mathbb{N}(n)$ to be a distribution on the natural numbers less
than or equal to $n$ such that, for $0 \leq i<n$, 
\[
\Pr_{x\sim D_\mathbb{N}(n)} \left[ x = i\right] = \Pr_{y\sim D_\mathbb{N}} \left[ y = i\right]
\]
and 
\[
\Pr_{x\sim D_\mathbb{N}(n)} \left[ x = n\right] = \Pr_{y\sim D_\mathbb{N}} \left[ y \geq n\right].
\]
That is, we are ``piling up'' the tail of the $D_\mathbb{N}$ distribution
on $n$. 

The following theorem shows that we can use $D_\mathbb{N}(n)$ to bound
the number of comparisons in an ECS algorithm
when the equivalence classes are drawn from $D$.
In particular, we focus here on 
an algorithm by Jayapaul {\it et al.}~\cite{Jayapaul2015} for equivalence
class sorting,
which involves a round-robin testing regiment, such
that each element, $x$, initiates
a comparison with the next element, $y$, with an unknown relationship to $x$,
until all equivalence classes are known.

\begin{theorem}\label{thm:dist-runtime}
Given a distribution, $D$, on a set of equivalence classes, 
then $n$ elements who have corresponding equivalence class 
independently drawn from $D$ can be equivalence class sorted using 
a total number of comparisons stochastically dominated by twice the sum of $n$ draws 
from the distribution $D_\mathbb{N}(n)$.
\end{theorem}
\begin{proof}
Let $V_i$ denote the random variable that is equal to the natural number
corresponding to the equivalence class of element $i$
in $D_\mathbb{N}(n)$.  
We denote the number of 
elements in equivalence class $i$ as $Y_i$.  
Let us denote
the number of equivalence tests performed by 
the algorithm 
by Jayapaul {\it et al.}~\cite{Jayapaul2015}
using the random variable, $R$.  

By a lemma from~\cite{Jayapaul2015}, 
for any pair of equivalence classes, $i$ and $j$,
the round-robin ECS algorithm 
performs at most $2 \min(Y_i,Y_j)$ equivalence tests in total.  
Thus, the total number of equivalence tests in our distribution-based
analysis is upper bounded by 

\begin{eqnarray*}
R & \leq & \sum_{i=0}^\infty \sum_{j=0}^{i-1} 2\min(Y_i,Y_j)\\
 & = & 2\sum_{i=0}^n \sum_{j=0}^{i-1} \min(Y_i,Y_j) 
     + 2\sum_{i=n+1}^\infty \sum_{j=0}^{i-1} \min(Y_i,Y_j)\\
  &\leq& 2\sum_{i=0}^n \sum_{j=0}^{i-1} Y_i  + 2\sum_{i=n+1}^\infty nY_i\\
  & \leq & 2\left(\sum_{i=0}^n i Y_i + \sum_{i=n+1}^\infty n Y_i\right) = 2 \sum_{i=1}^{n} V_i
\end{eqnarray*}

The second line in the above simplification 
is a simple separation of the double summation
and the third line follows because $\sum_{j=0}^{i-1} \min(Y_i,Y_j)$ is zero 
if $Y_i$ is zero and at most $n$, otherwise.  
So the total number of comparisons in the algorithm is 
bounded by twice the sum of $n$ draws from $D_\mathbb{N}(n)$.  
\end{proof}

Given this theorem, we can apply it to a number of distributions to 
show that the total number of comparisons performed is linear with
high probability.

\begin{theorem}
If $D$ is a discrete uniform, a geometric, or a Poisson distribution on a set equivalence classes, then it is possible to equivalence class sort 
using linear total number of comparisons 
with exponentially high probability.
\end{theorem}
\begin{proof}
The sum of $n$ draws from 
$D_\mathbb{N}(n)$ is stochastically dominated 
by the sum of $n$ draws from $D_\mathbb{N}$.
Let us consider each distribution in turn.

\begin{itemize}
\item
Uniform:
The sum of $n$ draws from a discrete uniform distribution 
is bounded by $n$ times the maximum value.
\item
Geometric:
Let $p$ be the parameter of a geometric distribution and let 
$X = \sum_{i=0}^{n-1} X_i$ where the $X_i$ are drawn from $\mathop{Geom}(p)$,
which is, of course, related to the Binomial distribution, $\mathop{Bin}(n,p)$,
where one flips $n$ coins with probability $p$ and records 
the number of ``heads.''
Then, by a Chernoff bound 
for the geometric distribution (e.g., see~\cite{mitzenmacher2005probability}),
\begin{eqnarray*}
\Pr[X - (1/p)n > k] & = & \Pr[\mathop{Bin}(k + (1/p)n,p) < n] \\
                   & \leq & e^{-2\frac{(pk + n - n)^2}{k + (1/p)n}}\\
\Pr[X > (2/p)n] & \leq & e^{-np}
\end{eqnarray*}
\item
Poisson:
Let $\lambda$ be the parameter of a Poisson distribution and let $Y = \sum_{i=0}^{n-1} Y_i$ where the $Y_i$ are drawn from $\mathop{Poisson}(\lambda)$.
Then, by a Chernoff bound 
for the Poisson distribution (e.g., see~\cite{mitzenmacher2005probability}),
\begin{eqnarray*}
\Pr[Y > (\lambda (e-1) + 1) n] & = & \Pr[e^Y > e^{(\lambda (e-1) + 1)n}] \\
& \leq & \frac{(E[e^{Y_i}])^n}{e^{(\lambda (e-1) + 1)n}} \\
 & = & \frac{e^{\lambda(e - 1)n}}{e^{(\lambda (e-1) + 1)n}} = e^{-n}
\end{eqnarray*}
\end{itemize}
So, in each case with exponentially high probability, the sum of $n$ draws 
from the distribution is $O(n)$ and the round-robin 
algorithm does $O(n)$ total equivalence tests.
\end{proof}

We next address the zeta distribution.

\begin{theorem}
Given a zeta distribution with parameter $s>2$, $n$ elements who have corresponding equivalence class independently drawn from the zeta distribution can be equivalence class sorted in $O(n)$ work in expectation.
\end{theorem}
\begin{proof}
When $s > 2$, the mean of the zeta distribution is 
\[
\frac{\zeta(s-1)}{\zeta(s)},
\]
which is a constant.
So the sum of $n$ draws from the distribution is expected to be linear.
Therefore, the expected total number of 
comparisons in the round-robin algorithm is linear.
\end{proof}

Unfortunately, for zeta distributions it is not immediately
clear if it is possible to improve the above theorem
so that total number of comparisons is shown to be linear 
when $2 \geq s > 1$ or obtain high probability bounds on these bounds.  
This uncertainty 
motivates us to look experimentally at how 
different values of $s$ cause the runtime to behave.
Likewise, our high-probability bounds on the total number 
of comparisons in the round-robin algorithm for the other distibutions
invites experimental analysis as well.

\section{Experiments} \label{sec:sort-exper}

In this section, we report on experimental validatations of 
the theorems from the
previous section and investigations of the behavior of running 
the round-robin algorithm on the
zeta distribution.  For the uniform, geometric, and Poisson distributions,
we ran ten tests on sizes of $10,000$ to $200,000$ elements
incrementing in steps of $10,000$.  For the zeta distribution,
because setting $s < 2$ seems to lead to a super linear number of
comparisons, we reduced the test sizes by a factor of $10$ and ran
ten tests each on sizes from $1,000$ to $20,000$ in increments of
$1,000$.  For each distribution we used the following parameter
settings for various experiments:

\begin{center}
\begin{tabular}{ l l }
Uniform: & $k = 10,25,100$\\
Geometric: & $p = \frac{1}{2},\frac{1}{10},\frac{1}{50}$\\
Poisson: & $\lambda = 1,5,25$\\
Zeta: & $s = 1.1,1.5,2,2.5$
\end{tabular}
\end{center}

The results of these tests are plotted in Figure~\ref{fig:exper-results}.  
Best fit lines were fitted whenever we have theorems
stating that there will be a linear number of comparisons with 
high probability or in expectation (i.e., everything except for 
zeta with $s < 2$).  
We include extra plots of the zeta distribution tests 
with the $s=1.1$ data and 
the $s = 1.1,1.5$ data removed to better see the other data sets.

\begin{figure*}[t]
\centering
\begin{subfigure}{.45\textwidth}
\centering
\includegraphics[scale=0.35]{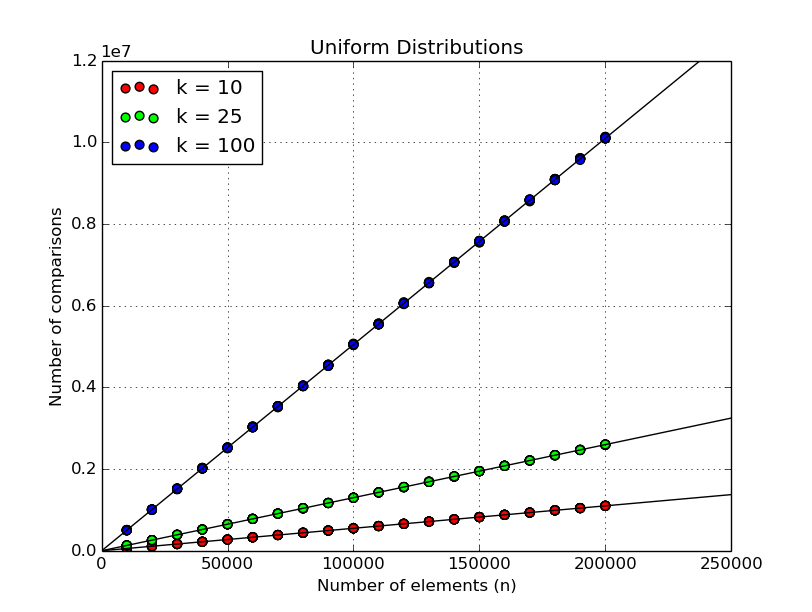}
\end{subfigure}
\begin{subfigure}{.45\textwidth}
\centering
\includegraphics[scale=0.35]{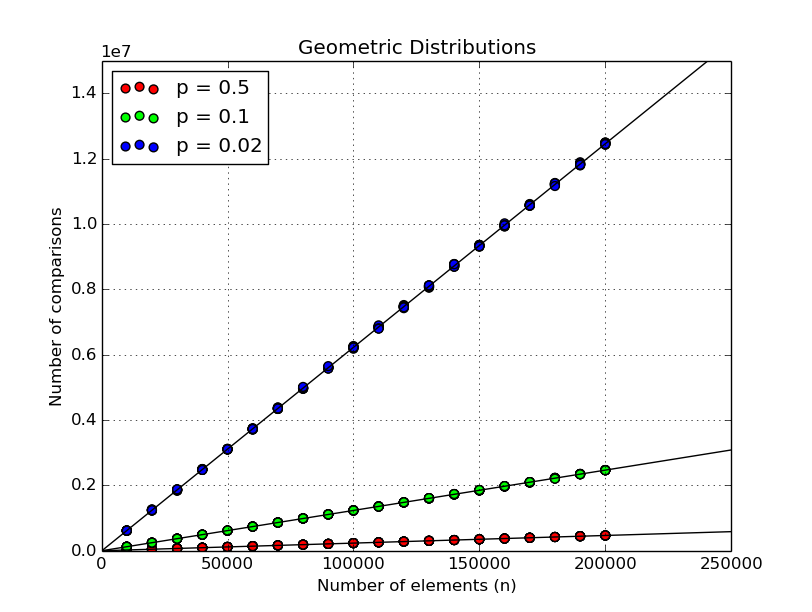}
\end{subfigure}
\begin{subfigure}{.45\textwidth}
\centering
\includegraphics[scale=0.35]{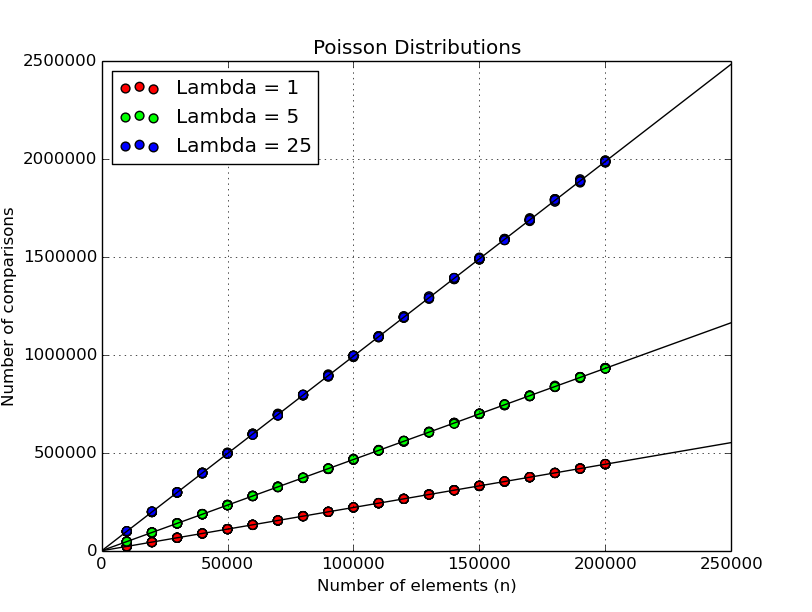}
\end{subfigure}
\begin{subfigure}{.45\textwidth}
\centering
\includegraphics[scale=0.35]{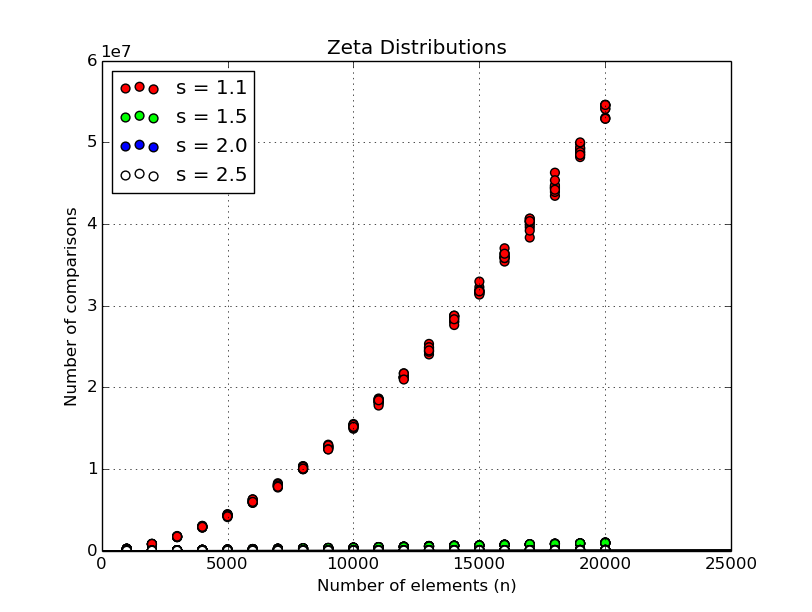}
\end{subfigure}
\begin{subfigure}{.45\textwidth}
\centering
\includegraphics[scale=0.35]{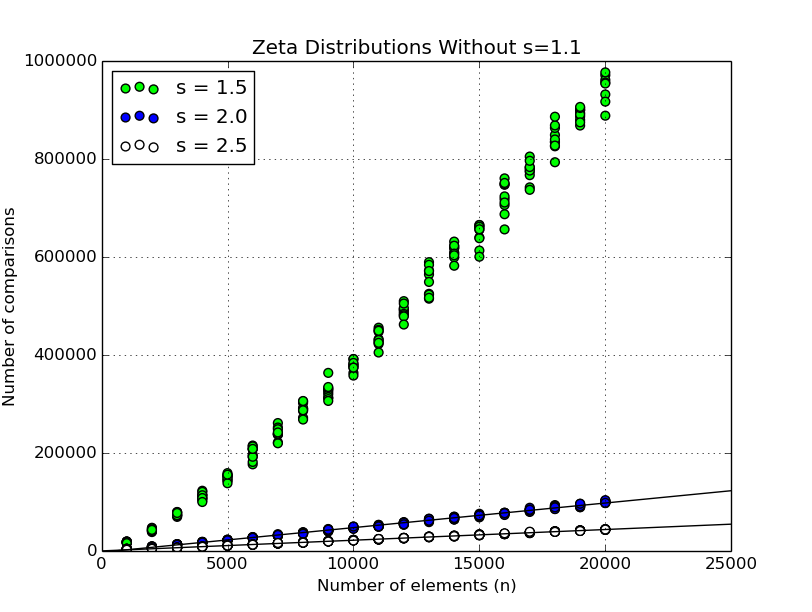}
\end{subfigure}
\begin{subfigure}{.45\textwidth}
\centering
\includegraphics[scale=0.35]{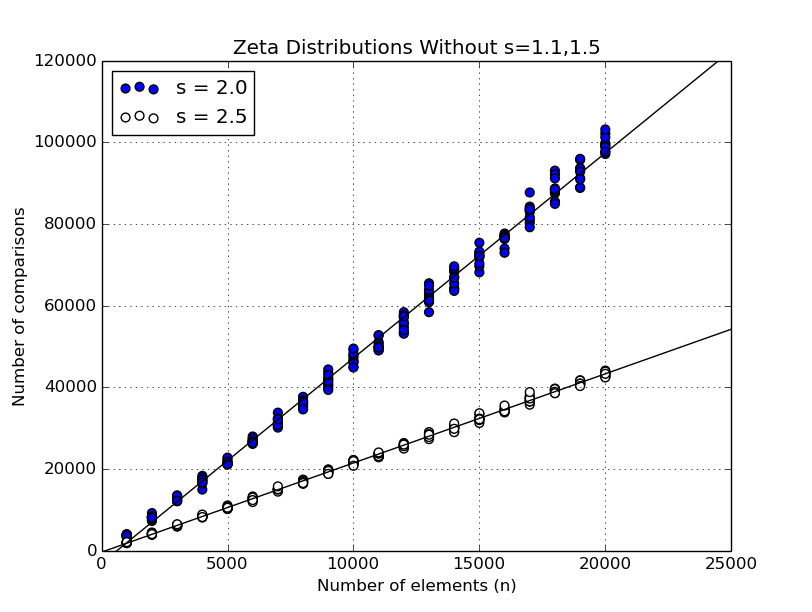}
\end{subfigure}
\caption{The results of the experiments are plotted and best fit lines are placed when we have a linear number of comparisons with high probability or in expectation.}
\label{fig:exper-results}
\end{figure*}

We can see from the data that the number of comparisons for 
the uniform, geometric, and Poisson distributions are so tightly concentrated 
around the best fit line that only one data point is visible.  
Contrariwise, the data points for the zeta distributions do not cluster nearly as nicely.  
Even when we have linear expected comparisons with $s=2$, the data points 
vary by as much as $10\%$.  

\section{Conclusion}
In this paper we have 
studied the equivalence class sorting problem, 
from a parallel perspective,
giving several new algorithms,
as well as new lower bounds and distribution-based analysis.
We leave as open problems the following interesting questions:
\begin{itemize}
\item
Is it possible to find all equivalance
classes in the ER version of the ECS problem
in $O(k)$ parallel rounds, for $k\ge 3$, where $k$ is the number
of equivalence classes?
Note that the answer is ``yes'' for $k=2$, as it follows from previous
results for the parallel fault diagnosis
problem~\cite{Beigel:1989,Beigel:1993,b492587}.
\item
Is it possible to bound the number of comparisons away from 
$O(n^2)$ for the zeta distribution when $s<2$ even just in expectation?
\item
Is it possible to prove a high-probability concentration bound for
the zeta distribution, similar to the concentration bounds we proved
for other distributions?
\end{itemize}

\subsection*{Acknowledgments}
This research was supported in part by
the National Science Foundation under grant 1228639
and a gift from the 3M Corporation.
We would like to thank David Eppstein and Ian Munro
for several helpful discussions
concerning the topics of this paper.

{\raggedright
\bibliographystyle{abbrv}
\bibliography{refs}
}

\end{document}